\documentclass[11pt]{article}
\usepackage[utf8]{inputenc}
\usepackage{fullpage}
\usepackage{graphicx}
\usepackage{amsthm}
\usepackage{amsmath}
\usepackage{amsfonts}
\usepackage{amssymb}

\usepackage{caption}
\usepackage{subcaption}
\usepackage{qcircuit}
\usepackage{tikz}
\usepackage{thm-restate}
\usepackage{mathtools}

\usepackage[pagebackref,colorlinks,citecolor=blue,,linkcolor=blue,bookmarks=true]{hyperref}

\renewcommand{\backref}[1]{}

\renewcommand{\backrefalt}[4]{%
\ifcase #1 %
\or
[p.\ #2]%
\else
[pp.\ #2]%
\fi}

\usepackage[nameinlink,capitalize]{cleveref}

\newcommand{\bitz}{\{0,1\}}
\newcommand{\eps}{\varepsilon}

\newcommand{\abs}[1]{\lvert #1 \rvert}

\newtheorem{theorem}{Theorem}
\newtheorem{definition}[theorem]{Definition}

\newtheorem{lemma}[theorem]{Lemma}
\newtheorem{fact}[theorem]{Fact}

\newcommand{\ket}[1]{\vert #1 \rangle}
\newcommand{\bin}[0]{\{0,1\}}

\renewcommand{\Pr}{\mathop{\bf Pr\/}}

\newcommand{\calL}{\mathcal{L}}
\newcommand{\calH}{\mathcal{H}}

\newcommand{\norm}[1]{\lVert #1 \rVert}
\newcommand{\N}{\mathbb{N}}

\newcommand{\calT}{\mathcal{T}}
\newcommand{\EXP}{\mathsf{EXP}}
\newcommand{\BQP}{\mathsf{BQP}}
\newcommand{\DQP}{\mathsf{DQP}}
\newcommand{\PDQP}{\mathsf{PDQP}}
\newcommand{\PDQMA}{\mathsf{PDQMA}}
\newcommand{\DQMA}{\mathsf{DQMA}}
\newcommand{\NEXP}{\mathsf{NEXP}}
\newcommand{\NP}{\mathsf{NP}}
\newcommand{\QMA}{\mathsf{QMA}}
\newcommand{\MIP}{\mathsf{MIP}}
\newcommand{\PCP}{\mathsf{PCP}}
\newcommand{\ALL}{\mathsf{ALL}}
\newcommand{\F}{\mathbb{F}}

\newcommand{\verifier}{V}

\newcommand{\pcpverifier}{V_{\mathsf{PCP}}}
\newcommand{\poly}{\mathrm{poly}}

\title{
$\PDQMA = \DQMA = \NEXP$:\\$\QMA$ With Hidden Variables and Non-collapsing Measurements 
}

\author{
Scott Aaronson\thanks{University of Texas at Austin. \href{mailto:aaronson@cs.utexas.edu}{\texttt{aaronson@cs.utexas.edu}}.}
\and
Sabee Grewal\thanks{University of Texas at Austin. \href{mailto:sabee@cs.utexas.edu}{\texttt{sabee@cs.utexas.edu}}.}
\and
Vishnu Iyer\thanks{University of Texas at Austin. \href{mailto:vishnu.iyer@utexas.edu}{\texttt{vishnu.iyer@utexas.edu}}. }
\and
Simon C. Marshall\thanks{Leiden University. \href{mailto:s.c.marshall@liacs.leidenuniv.nl}{\texttt{s.c.marshall@liacs.leidenuniv.nl}}. }
\and
Ronak Ramachandran\thanks{University of Texas at Austin. \href{mailto:ronakr@utexas.edu}{\texttt{ronakr@utexas.edu}}.}
}

\date{}

\begin{document}

\maketitle

\begin{abstract}
We define and study a variant of $\mathsf{QMA}$ (Quantum Merlin Arthur) in which Arthur can make multiple non-collapsing measurements to Merlin's witness state, in addition to ordinary collapsing measurements. \ By analogy to the class $\mathsf{PDQP}$ defined by Aaronson, Bouland, Fitzsimons, and Lee (2014), we call this class $\mathsf{PDQMA}$. \ Our main result is that $\mathsf{PDQMA}=\mathsf{NEXP}$; this result builds on the $\PCP$ theorem and complements the result of Aaronson (2018) that $\mathsf{PDQP/qpoly} = \mathsf{ALL}$. \ While the result has little to do with quantum mechanics, we also show a more ``quantum'' result: namely, that $\mathsf{QMA}$ with the ability to inspect the entire history of a hidden variable is equal to $\mathsf{NEXP}$, under mild assumptions on the hidden-variable theory. \
We also observe that a quantum computer, augmented with quantum advice and the ability to inspect the history of a hidden variable, can solve any decision problem in polynomial time.
\end{abstract}
\newpage

\section{Introduction}

To understand the power of quantum computation is, in large part, to understand how that power depends on the central features of quantum mechanics itself, such as linearity, unitarity, tensor products, complex numbers, the Born Rule, or the destructive nature of measurement. \ But since quantum mechanics is usually presented as a ``package deal,'' how can we pick apart these dependencies? \ One natural approach has been to define complexity classes based on ``fantasy'' versions of quantum mechanics, which change one or more of its features, and see how they relate to the standard $\mathsf{BQP}$ (Bounded-Error Quantum Polynomial-Time). \ Some representative findings of that research program over the past few decades include:

\begin{enumerate}
\item[(1)] Quantum mechanics over the reals or quaternions leads to the same computational power as quantum mechanics over the complex numbers, despite the theories differing in other respects \cite{rudolph20022}.

\item[(2)] Quantum mechanics with a nonlinear Schr\"{o}dinger equation would generically allow $\mathsf{NP}$- and even $\mathsf{\#P}$-complete problems to be solved in polynomial time, in contrast to what is conjectured for standard (linear) quantum mechanics \cite{al}.

\item[(3)] A quantum computer with closed timelike curves could solve exactly the problems in $\mathsf{PSPACE}$, same as a classical computer with closed timelike curves \cite{awat}.

\item[(4)] Quantum computers with nonunitary linear evolution, or modifications of the Born rule (say, $|\psi|^3$), with normalization of probabilities imposed, yield at least the power of quantum computers with \textit{postselected measurement outcomes}---a model that Aaronson called $\mathsf{PostBQP}$ and proved to coincide with the classical complexity class $\mathsf{PP}$ \cite{aar:isl,aar:pp}.

\item[(5)] Generalized probabilistic theories (GPTs), a class of theories that includes quantum mechanics and classical probability theory as special cases, are as a whole characterized by the complexity class $\mathsf{AWPP}$ \cite{lee2015computation, barrett2019computational}.

\item[(6)] If we allow multiple, non-collapsing measurements of the same state---or, closely related, to see the entire history of a hidden variable as in Bohmian mechanics---we get a model of computation that seems more powerful than standard quantum computation, but only ``slightly'' so \cite{aar:qchv}. \ As examples, we can quickly find collisions in many-to-one functions (and thus, for example, solve Graph Isomorphism), and we can solve the Grover search problem in $N^{1/3}$ steps rather than $\sqrt{N}$. \ But we still seem unable to solve $\mathsf{NP}$-complete problems in polynomial time.
\end{enumerate}

Example (6) is the one of most interest to us here. \ To our knowledge, it is the only natural example known where changing the rules of quantum mechanics leads to complexity classes that appear only modestly larger than $\mathsf{BQP}$. \ Much of the power comes from the combination of non-collapsing measurements with ordinary collapsing ones. \ As an example, given a two-to-one function $f:[N]\rightarrow[M]$, the way to find collisions is simply to prepare
$$ \frac{1}{\sqrt{N}} \sum_{x\in [N]} \ket{x} \ket{f(x)}, $$
then measure the $\ket{f(x)}$ register in the ordinary way to get
$$ \frac{\ket{x} + \ket{y}}{\sqrt{2}} $$
where $f(x)=f(y)$ in the first register, and finally perform multiple non-collapsing measurements on the first register in the standard basis, until both $x$ and $y$ are observed with high probability.\footnote{The reason why non-collapsing measurement allows Grover search in $N^{1/3}$ steps is simpler and does not require us to combine collapsing with non-collapsing measurements. \ Instead, given a unique marked item out of $N$, one simply runs Grover's algorithm for $T=N^{1/3}$ iterations, thereby boosting the probability of the marked item to $\sim \frac{T^2}{N} = N^{-1/3}$, and then performs $N^{1/3}$ non-collapsing measurements so that the marked item is found with constant probability.}

In a hidden-variable theory, where the hidden variable is either at $\ket{x}$ or $\ket{y}$ with equal probabilities, the solution at this point is to ``juggle''---for example, by repeatedly applying Fourier transforms followed by their inverses. \ The goal here is to cause the hidden variable to ``forget'' whether it was at $\ket{x}$ or $\ket{y}$, so that it must eventually visit both of them with high probability. \ In such a case, if (as we're imagining) we could see the whole history of the hidden variable at once, from some godlike vantage point, we would learn both $\ket{x}$ and $\ket{y}$ and thereby solve our computational problem.

The history of these ideas is a bit tangled. \ In 2005, Aaronson defined the class $\mathsf{DQP}$ (Dynamical Quantum Polynomial-Time), to capture the problems that quantum computers could efficiently solve, if only one could examine the entire history of a hidden variable (in any hidden-variable theory that satisfies reasonable axioms, called robustness and indifference). \ He showed that $\mathsf{SZK}\subseteq \mathsf{DQP}$, where $\mathsf{SZK}$ is Statistical Zero Knowledge, basically because $\mathsf{DQP}$ could simulate non-collapsing measurements (although he didn't formalize this). \ Combined with Aaronson's quantum lower bound for finding collisions \cite{aar:col}, which implies the existence of an oracle relative to which $\mathsf{SZK} \not\subset \mathsf{BQP}$, this gives us an oracle separation between $\mathsf{BQP}$ and $\mathsf{DQP}$. \ Aaronson also showed that $\mathsf{DQP} \subseteq \mathsf{EXP}$, which has not been improved since then. \ He claimed to give an oracle relative to which $\mathsf{NP} \not\subset \mathsf{DQP}$, although his proof had a bug \cite{abfl} and the existence of such an oracle remains open. \ Then, in 2014, Aaronson, Bouland, Fitzsimons, and Lee \cite{abfl} defined the class $\mathsf{PDQP}$ (Product Dynamical Quantum Polynomial-Time), to capture non-collapsing measurements specifically. \ They showed that $\mathsf{PDQP}$ also contains $\mathsf{SZK}$, showed the upper bound $\mathsf{PDQP} \subseteq \mathsf{BPP}^{\mathsf{PP}}$, and gave a correct proof that there exists an oracle relative to which $\mathsf{NP} \not\subset \mathsf{PDQP}$.

Overall, as we said, these results painted a picture of $\mathsf{DQP}$ and $\mathsf{PDQP}$ as only modestly more powerful than $\mathsf{BQP}$. \ However, a surprising new wrinkle came in 2018, when Aaronson \cite{aar:pdqpqpoly} observed that $\mathsf{PDQP/qpoly} = \mathsf{ALL}$, where $\mathsf{/qpoly}$ means ``with polynomial-sized quantum advice,'' and $\mathsf{ALL}$ is the class of all languages. \ This stands in contrast to 2004 results of Aaronson \cite{aar:adv} that limit the power of $\mathsf{BQP/qpoly}$: namely, that $\mathsf{BQP/qpoly} \subseteq \mathsf{PP/poly}$ (later improved by Aaronson and Drucker \cite{adrucker} to $\mathsf{BQP/qpoly} \subseteq \mathsf{QMA/poly}$), and that there exists an oracle relative to which $\mathsf{NP} \not\subset \mathsf{BQP/qpoly}$. \ In other words, quantum advice and non-collapsing measurements have a ``Mentos and Coke'' character, where each one is only modestly powerful in isolation, but together they trigger a complexity-theoretic explosion.

To prove the $\mathsf{PDQP/qpoly} = \mathsf{ALL}$ result, Aaronson adapted a 2005 theorem of Raz \cite{raz:all} that $\mathsf{QIP/qpoly} = \mathsf{ALL}$, where $\mathsf{QIP}$ is the class of languages that admit quantum interactive proofs. \ In both Raz's protocol and Aaronson's, given an arbitrary Boolean function $f:\bin^n\rightarrow\bin$ that one wants to compute, one first chooses a prime $q \gg n$. \ The whole truth table of $f$ is then encoded by the quantum advice state
$$ \ket{\psi} = \frac{1}{\sqrt{q^n}} \sum_{z\in\mathbb{F}_q^n} \ket{z}\ket{p(z)}, $$
where $p:\mathbb{F}_q^n\rightarrow \mathbb{F}_q$ is the unique multilinear polynomial over $\mathbb{F}_q$ such that $p(x)=f(x)$ for all $x\in\bin^n$. \ Next, given a point of interest $x\in\bin^n$, on which one wants to evaluate $f(x)$, one measures $\ket{\psi}$ so as to collapse it to an equal superposition
$$\ket{\psi_\ell} = \frac{1}{\sqrt{q-1}} \sum_{z\in \ell \setminus \{x\} } \ket{z}\ket{p(z)} $$
over a random line $\ell \subset \mathbb{F}_q^n$ that passes through $x$, minus the point $x$ itself. \ Finally, one uses polynomial interpolation on this line to recover $p(x)=f(x)$. \ In the non-collapsing measurements model, this is done by simply measuring the state $\ket{\psi_\ell}$ over and over in the standard basis, until enough $(z,p(z))$ pairs have been observed for the interpolation to work.\footnote{Since $p$ is a multilinear extension of a Boolean function on $n$ variables, its degree is at most $n$. \ Hence, $n+1$ pairs are needed to do polynomial interpolation, so $q$ must be chosen to be at least $n+2$.}

\subsection{This Paper}

Here we show a new example where non-collapsing measurements combined with one other resource yield extraordinary computational power---vastly more power than either resource in isolation.

The class $\mathsf{QMA}$ (Quantum Merlin Arthur) is a well-known quantum generalization of $\mathsf{NP}$; it consists of all languages for which a ``yes'' answer can be verified in quantum polynomial time with the help of a polynomial-size quantum witness state. \ We define and study $\mathsf{PDQMA}$ (Product Dynamical $\mathsf{QMA}$), or $\mathsf{QMA}$ augmented with non-collapsing measurements. \ Our main result is that $\mathsf{PDQMA} = \mathsf{NEXP}$, even if $\mathsf{PDQMA}$ is defined with a completeness/soundness gap of (say) $1-2^{-n}$ vs.\ $2^{-n}$ (\cref{thm:pdqma=nexp}). 

Since the inclusion $\mathsf{PDQMA} \subseteq \mathsf{NEXP}$ is straightforward, the interesting part is $\mathsf{NEXP} \subseteq \mathsf{PDQMA}$. \ 
By the celebrated PCP theorem \cite{arora1998probabilistic,arora1998proofverification}, it suffices to show that a two-query PCP system (\cref{def:pcp}) can be simulated by $\PDQMA$. \
Our proof builds on Aaronson's proof \cite{aar:pdqpqpoly} that $\mathsf{PDQP/qpoly}  = \ALL$, but with two key differences. \ First, to simulate a two-query PCP system, we need a witness state that can support at least \emph{two} queries to an exponentially long truth table, rather than just one query. \  
Second, we now need an analysis of \textit{soundness}: why, for example, can the prover not cheat by sending a witness that causes the response to each query to depend on the other query? \ This problem seems particularly acute once we realize that, when non-collapsing measurements are allowed, we can no longer rely on the no-signaling principle of quantum mechanics. \ Fortunately, there turns out to be a simple solution: basically, the ``error-correcting code'' structure of an honest witness state prevents a cheating prover from gaining anything by correlating the query responses.

We point out some implications of our result. \ First, combining with the Nondeterministic Time Hierarchy Theorem ($\mathsf{NP} \ne \mathsf{NEXP}$), we find that $\mathsf{PDQMA}$ is \textit{unconditionally} more powerful than $\mathsf{NP}$. \ Second, ``scaling down by an exponential,'' we find that when non-collapsing measurements are allowed, the problem of optimizing an acceptance probability over all $N$-dimensional quantum states jumps from easy (a principal eigenvector problem) to $\mathsf{NP}$-hard even to approximate.

Let us place our result in the context of previous work on generalizations of $\mathsf{QMA}$. \ Aharonov and Regev \cite{ar:quantum} defined $\mathsf{QMA+}$, a variant of $\QMA$ where the verifier can directly obtain the probability a given two-outcome measurement will accept, and showed $\mathsf{QMA+} = \QMA$. \ More recently, Jeronimo and Wu \cite{jeronimo:nonnegative} showed that $\mathsf{QMA^+(2)} = \mathsf{NEXP}$, where $\mathsf{QMA^+(2)}$ is $\mathsf{QMA}$ with two unentangled proofs that have nonnegative real amplitudes. \ Bassirian, Fefferman, and Marwaha \cite{bassirian2023quantum} improved this to show that $\mathsf{QMA^+} = \mathsf{NEXP}$---i.e., nonnegative amplitudes alone suffice for the jump to $\mathsf{NEXP}$. \ As of this writing, it remains a matter of avid speculation whether unentangled witnesses \textit{also} suffice for the jump to $\mathsf{NEXP}$: that is, whether $\mathsf{QMA(2)} = \mathsf{NEXP}$. \ Of course, our result implies that it would suffice to simulate non-collapsing measurements using unentangled proofs---i.e., to show $\mathsf{PDQMA} \subseteq \mathsf{QMA(2)}$.

One important observation about our $\mathsf{PDQMA}$ protocol is that it only ever measures the witness state in the computational basis---and hence, one could say, never exploits quantum interference. \ So in particular, if we defined a complexity class $\mathsf{PDMA}$ (Product Dynamical Merlin-Arthur) in mathematical parallel to $\mathsf{PDQMA}$, where the witness was a classical probability distribution $\mathcal{D}$, and one was allowed to sample from $\mathcal{D}$ with or without doing Bayesian updating to it, we would equally have $\mathsf{PDMA} = \mathsf{NEXP}$. \ The main difference here is simply that it seems hard to invent a story that motivates $\mathsf{PDMA}$.

In \cref{sec:dqma}, we sharpen this point, by defining and studying a class that we call $\mathsf{DQMA}$ (Dynamical $\mathsf{QMA}$), or $\mathsf{QMA}$ augmented by the ability to see the entire history of a hidden variable. \ We show that $\mathsf{DQMA}=\mathsf{NEXP}$ (\cref{thm:dqma=nexp}). \ Here the reasons really do depend on quantum mechanics. \ Specifically, they depend on the ability to ``hide'' crucial information in the phases of amplitudes, to prevent a hidden variable trajectory from remembering that information. \ If we tried to define the analogous class $\mathsf{DMA}$ (Dynamical $\mathsf{MA}$), it would trivially coincide with $\mathsf{MA}$. \ Assuming $\mathsf{MA}=\mathsf{NP}$, as follows from a standard derandomization assumption, this has the amusing consequence that $\mathsf{DMA}\ne \mathsf{DQMA}$---that is, in the presence of both witness states and trajectory sampling, quantum can already be known to be stronger than classical. \
In the same section, we also observe that our techniques imply $\mathsf{DQP/qpoly} = \mathsf{ALL}$ (\cref{thm:dqp/qpoly=all}), complementing Aaronson's result that $\mathsf{PDQP/qpoly} = \mathsf{ALL}$ \cite{aar:pdqpqpoly}. \ 
This result relies on quantum mechanics in the same way that our $\DQMA$ result does.

One might also wonder about other ``fantasy'' variants of $\QMA$. \ 
In principle, any variant of $\BQP$ can be combined with $\QMA$ to yield a new complexity class. \ 
For example, consider a variant of $\BQP$ that can clone states (i.e., perform the transformation $\ket\psi\ket{0^n} \mapsto \ket\psi^{\otimes 2}$). \
It is easy to see that one can simulate $k$ non-collapsing measurements of a state by cloning the state $k$ times and then measuring each copy in the usual collapsing way. \ 
Hence, combining this $\BQP$ variant with $\QMA$ yields a complexity class equal to $\NEXP$ by our \cref{thm:pdqma=nexp}. \
Indeed, if one wants to alter Arthur's powers \emph{without} triggering a ``Mentos and Coke'' effect, they would need to find a $\BQP$ variant that is only modestly more powerful, the way $\PDQP$ is. \

\paragraph{Concurrent Work}
In concurrent and independent work, Bassirian and Marwaha \cite{bassirian2024efficient} show that $\QMA$ with non-collapsing measurements equals $\NEXP$.  
They observe that the proof of $\QMA^+ = \NEXP$ \cite{bassirian2023quantum} goes through even if one replaces the promise of a non-negative witness state with the ability for the verifier to perform non-collapsing measurements. 
\cite{bassirian2023quantum} can prove a containment (for a constant completeness/soundness gap) using a constant number of non-collapsing measurements, whereas our verification procedure always uses $O(n \log n)$ non-collapsing measurements. 
Our approach readily extends to prove $\DQMA = \NEXP$ and $\DQP\mathsf{/qpoly} = \ALL$.

\subsection{Main Ideas}
We give a high-level overview of our proof that $\PDQMA = \NEXP$ (\cref{thm:pdqma=nexp}). \ 
Informally, $\PDQMA$ is like $\QMA$ except the verifier can perform non-collapsing measurements in addition to the normal collapsing ones (see \cref{def:pdqma} for a formal definition). \
The containment $\PDQMA \subseteq \NEXP$ is straightforward because $\NEXP$ can guess the exponentially-long classical description of the $\PDQMA$ witness and verify it in exponential time. \ 

Thus, the challenge is proving $\NEXP \subseteq \PDQMA$. \ 
We show this by simulating a two-query probabilistically checkable proof (PCP) system for $\NEXP$ (see \cref{def:pcp} and \cref{mipnexp}). \ 
In short, the celebrated PCP theorem \cite{arora1998probabilistic,arora1998proofverification} tells us that languages in $\NEXP$ can be decided by a verifier that queries a PCP $\pi : \bitz^n \to \Sigma$ at two points of the verifier's choosing, where $\Sigma$ is a constant-sized alphabet. \
Hence, our result follows from simulating two queries to a $\pi$. 

The honest witness is the same as in \cite{aar:pdqpqpoly}:
\[
\ket\psi = \frac{1}{\sqrt{q^n}} \sum_{z \in \F_q^n} \ket{z}\ket{p(z)},
\]
where $p:\F_q^n \to \F_q$ is the unique degree-$n$ multilinear extension of the PCP $\pi: \bitz^n \to \Sigma$ and $q = O(n)$ is chosen to be sufficiently large. \
However, our setting differs from \cite{aar:pdqpqpoly} in two key ways: (i) we must retrieve two values $\pi(w)$ and $\pi(w')$ for our choice of $w$, $w'$, rather than one, and (ii) the quantum witness can no longer be trusted (as it can be in the advice setting). 

To explain how to handle the first difference, let us briefly recall how to simulate a single query.
As discussed, given a point of interest $w\in\bin^n$, on which one wants to evaluate $\pi(w)$, Aaronson \cite{aar:pdqpqpoly} measures $\ket{\psi}$ so as to collapse it to an equal superposition
\[ \ket{\psi_\ell} = \frac{1}{\sqrt{q-1}} \sum_{z\in \ell \setminus \{w\} } \ket{z}\ket{p(z)}\]
over a random line $\ell \subset \mathbb{F}_q^n$ that passes through $w$, minus the point $w$ itself. \ 
To recover two values $w,w' \in \{0,1\}^n$, one can generalize Aaronson's procedure so that the measurement on $\ket\psi$ collapses it to an equal superposition over an \emph{affine plane} that contains the points $w$ and $w'$, minus the unique affine line containing $w$ and $w'$. \ 
Indeed, this can be generalized to any $k$-dimensional affine subspace, where the measurement collapses $\ket\psi$ to a superposition over a random $k$-dimensional affine subspace containing points $w_1,\dots,w_k$ minus the unique $(k-1)$-dimensional affine subspace containing $w_1,\dots, w_k$. \
Then, just as in \cite{aar:pdqpqpoly}, one performs enough non-collapsing measurements to recover $\pi(w)$ and $\pi(w')$ via polynomial interpolation.

The only remaining issue is that Merlin can potentially cheat and send a quantum state that is far from the honest witness. \
To address this, we use the lines-point low-degree test (\cref{lemma:ldt}) given by Friedl and Sudan \cite{friedl1995some}. \
Recall that an affine line $\ell \subset \F_q^n$ passing through points $x,y \in \F_q^n$ is the set $\{x + (y-x)t\}_{t \in \F_q^n}$, so a function $g: \ell \to \F_q$ can be viewed as a univariate polynomial in $t$. \
The lines-point low-degree test says that if a function $f:\F_q^n \to \F_q$ restricted to a randomly chosen line agrees with a degree-$n$ univariate polynomial, then $f$ agrees with some degree-$n$ polynomial $h:\F_q^n \to \F_q$ on all of $\F_q^n$. 

Observe that to perform this low-degree test, we need to make sure the function that Merlin encodes into his witness agrees with a low-degree polynomial on a randomly chosen affine line. \
Therefore, rather than interpolating a polynomial on a random affine plane containing the points $w$ and $w'$ of interest, our verification procedure interpolates a polynomial on an \emph{affine cube}. This affine cube contains the points $w, w' \in \bitz^n$ but we also ensure that it contains a point $w'' \in \F_q^n$ that the verifier chooses uniformly at random. \ 
This ensures that the affine cube contains a uniformly random line independent of the points $w$ and $w'$. \
Then, if the polynomial interpolation succeeds, we can (i) recover $\pi(w)$ and $\pi(w')$ as desired and (ii) conclude that Merlin's witness encoded a low-degree polynomial because it passed the lines-point low degree test (i.e., it agreed with a low-degree polynomial on a randomly chosen line).

To summarize, the procedure works as follows. \ 
First, Merlin commits to some witness state $\ket\psi$. \ 
Then the $\PDQMA$ verifier simulates the $\PCP$ verifier to obtain queries $w, w' \in \bitz^n$ and picks a uniformly random $w'' \in \F_q^n$. \ 
The verifier measures Merlin's witness state $\ket\psi$ to collapse it to a superposition over points in a random affine cube containing $w, w'$, and $w''$ minus the affine plane containing $w$, $w'$, and $w''$. \
Then the verifier uses non-collapsing measurements to collect all $(z, p(z))$ pairs in the affine cube (minus the affine plane) and interpolates a polynomial $u: \F_q^3 \to \F_q$ that fits these pairs. \
If the polynomial interpolation succeeds, Merlin passes the lines-point low-degree test and the verifier can learn $\pi(w)$ and $\pi(w')$ by evaluating the polynomial $u$.
\section{Preliminaries}

We assume the reader is familiar with basic concepts in complexity theory and quantum computing. \ 
Throughout this work, we use the following notation. $\mathbb{N} \coloneqq \{1, 2, 3, \ldots\}$ denotes the natural numbers. For a finite field $\F_q$, $\F_q^* \coloneqq \F_q \setminus \{0\}$.

Let $\calL$ denote the set of all affine lines (i.e., $1$-dimensional affine subspaces) of $\F_q^n$.
Recall that the line passing through points $a, b \in \F_q^n$ is the set $\{a + (b-a)t\}_{t \in \F_q}$. 
Therefore, a polynomial $g$ that maps a line $\ell \in \calL$ to $\F_q$ can be thought of as a univariate polynomial.
We will use the following lines-point low-degree test of Friedl and Sudan \cite{friedl1995some} (see also the restatement in \cite[Theorem 1.2]{harsha2023improvedlinepointlowdegreetest}).

\begin{lemma}[Multivariate low-degree test \cite{friedl1995some}]\label{lemma:ldt}
   Let $f: \F_q^n \to \F_q$ be any $n$-variate function, and let $G = \{g_\ell\}_{\ell \in \calL}$ be a collection of degree-$d$ polynomials $g_\ell : \ell \to \F_q$. There is a constant $C$ large enough such that for any $d$ satisfying $q > Cd$, if  
   \[
   \Pr_{\substack{\ell \sim \calL\\ z\sim \ell}}[f(z) \neq g_\ell(z)] \leq \delta,
   \]
   for some $0 < \delta < 0.01$, then there exists an $n$-variate degree-$d$ polynomial $h$ such that 
   \[
   \Pr_{z \sim \F_q^n}[f(z) \neq h(z)] \leq 4 \delta.
   \]
\end{lemma}

In words, \cref{lemma:ldt} is saying that, if a function $f$ restricted to a randomly chosen line $\ell$ disagrees with a low-degree polynomial on at most a $\delta$ fraction of points, then there must exist a globally low-degree polynomial that $f$ disagrees with on at most a $4\delta$ fraction of points. In short, \cref{lemma:ldt} gives a local way to test that the entire function is low degree. 

This work also involves affine planes and cubes (i.e., $2$- and $3$-dimensional affine subspaces). 
Recall that an affine plane is uniquely defined by three independent points $a, b, c \in \F_q^n$. The plane containing these points is the set 
\[
\{a + (b-a)t_1 + (c - a)t_2\}_{t_1, t_2 \in \F_q}. 
\]
Similarly, the unique affine cube containing the independent points $a,b,c,d \in \F_q^n$ is the set 
\[
\{a + (b-a)t_1 + (c - a)t_2 + (d-a)t_3\}_{t_1, t_2,t_3 \in \F_q}. 
\]

Next, we give a formal definition of $\PDQMA$.

\begin{definition}\label{def:pdqma}
$\mathsf{PDQMA}_{c,s}$ is the class of languages $L\subseteq \bin^*$ for which there exists a $\mathsf{PDQP}$ verifier $\verifier$ (to be defined shortly) such that, for all $x\in\bin^*$:
\begin{itemize}
\item \textbf{Completeness:} If $x\in L$ then there exists a witness state $\ket{\phi}$, on $\operatorname{poly}(n)$ qubits, such that $V(x,\ket{\phi})$ accepts with probability at least $c$.
\item \textbf{Soundness:} If $x\not\in L$ then $V(x,\ket{\phi})$ accepts with probability at most $s$ for all witness states $\ket{\phi}$.
\end{itemize}
A $\mathsf{PDQP}$ verifier consists of two phases. \ In the first phase, a $\mathsf{P}$-uniform quantum circuit $C_x$, depending on the input $x$, is applied to the initial state $\ket{\phi}\ket{0^{p(n)}}$, where $\ket{\phi}$ is the witness and $p$ is a polynomial. \ This $C_x$ can consist of three types of gates: $\operatorname*{CNOT}$ gates, $1$-qubit $\pi/8$ rotations, and measurements in the $\{ \ket{0},\ket{1} \}$ basis. \ The $\operatorname*{CNOT}$ and $\pi/8$ gates provide universality for $\mathsf{BQP}$, while the measurement gates introduce a probabilistic component.
In the second phase, we consider the final state $\ket{\psi}$ of $C_x$, which depends in part on the probabilistic results of the measurement gates. \ Let $\mathcal D$ be the probability distribution induced by measuring all qubits of $\ket{\psi}$ in the computational basis. \ Then a classical polynomial-time algorithm, $A$, receives as input $x$ as well as $q(n)$ independent samples from $\mathcal D$, and then either accepts or rejects.
\end{definition}

Our proof that $\mathsf{PDQMA}=\mathsf{NEXP}$ will rely on the characterization $\NEXP$ by probabilistically checkable proof (PCP) systems where the verifier only makes $2$ queries to the PCP. 
Although this is well-known to classical complexity theorists, we explain this formally below for completeness.
We begin by defining the complexity classes $\PCP$.

\begin{definition}\label{def:pcp}%
    $\PCP_{c,s}[r,q]_{\Sigma}$ is the class of languages $L\subseteq \bin^*$ for which there exists a probabilistic polynomial-time verifier $\pcpverifier$ which uses $r$ random bits and makes $q$ queries to an oracle $\pi$, each time receiving a response in some alphabet $\Sigma$ such that
    \begin{enumerate}
        \item \textbf{Completeness:} If $x \in L$, then $\exists \pi$ such that  $\Pr[\pcpverifier^{\pi}(x) = 1] \geq c$.
        \item \textbf{Soundness:} If $x \not \in L$, then $\forall \pi$, $\Pr[\pcpverifier^{\pi}(x) = 1] \leq s$.
    \end{enumerate}
\end{definition}

The celebrated $\PCP$ theorem \cite{arora1998probabilistic, arora1998proofverification} tells us that $\NP = \PCP_{1,s}[O(\log(n)),3]_{\bin}$ for any constant $s$. 
A simple consequence is that $\NEXP = \PCP_{1,s}[\poly(n),3]_{\bin}$. 
It is folklore that $\PCP_{c,s}[r,q]_{\Sigma} \subseteq \PCP_{c,s/q}[r + \log(q),2]_{\Sigma^q}$, i.e., that the number of queries can be reduced to two at the cost of a larger alphabet and worse soundness error.

\begin{theorem} \label{mipnexp}
For a size-$8$ alphabet $\Sigma$ and any constant $s \in (0,1)$, $\PCP_{1,s}[\poly(n),2]_{\Sigma} = \NEXP$.
\end{theorem}

One can improve the soundness error to sub-constant by further increasing the alphabet size and our proofs that $\PDQMA = \NEXP$ (\cref{thm:pdqma=nexp}) and $\DQMA = \NEXP$ (\cref{thm:dqma=nexp}) still go through with only slight (if any) modification.

\section{\texorpdfstring{$\sf{QMA}$ and Non-collapsing Measurements}{QMA and Non-Collapsing Measurements}}\label{sec:pdqma}

We prove our main result: if $\QMA$ is modified so that Arthur can perform non-collapsing measurements in addition to standard quantum computation, then the resulting class equals $\NEXP$. 

The hard part is to show that $\NEXP \subseteq \PDQMA$, which we prove by simulating a two-query PCP system for $\NEXP$. \ 
At a high level, the $\PDQMA$ verifier is given a PCP $\pi: \bitz^n \to \Sigma$ encoded in a quantum proof, and it suffices to learn $\pi$ at two points of the verifier's choosing.

Our starting point is Aaronson's result that $\PDQP\mathsf{/qpoly} = \ALL$ \cite{aar:pdqpqpoly} where he showed how to evaluate $\pi$ at one point but this assumed that the verifier is provided with a \emph{trusted} quantum advice state. \ 
Hence, our contribution is to show that one can retrieve \emph{two} points of choice even if the quantum proof is from an \emph{untrusted prover}. 

\begin{theorem}\label{thm:pdqma=nexp}
$\mathsf{PDQMA} = \mathsf{NEXP}$.
\end{theorem}
\begin{proof}
$\mathsf{PDQMA} \subseteq \mathsf{NEXP}$ is clear, since in $\mathsf{NEXP}$ we can guess an exponentially-long classical description of the $\mathsf{PDQMA}$ witness and then verify it.

Thus, we show $\mathsf{NEXP} \subseteq \mathsf{PDQMA}$. \ By \cref{mipnexp}, it suffices to show $\PCP_{c,s}[\poly(n), 2]_\Sigma \subseteq \mathsf{PDQMA}$ for some constant-size alphabet $\Sigma$. \ In particular, we can assume the $\PCP$ verifier makes two queries to a PCP and receives responses in a constant-size alphabet. \ Let $\pi: \bin^n \rightarrow \Sigma$ be the Boolean function that encodes the PCP for all possible queries $x\in \bin^n$. \ 
Let $p : \F_q^n \to \F_q$ be the unique degree-$n$ multilinear extension of $\pi$, where $q$ is chosen so that the conditions of \cref{lemma:ldt} are satisfied and $q \geq n +2$. Note that $q = O(n)$ by Bertrand's postulate. \
To simulate the PCP system, it suffices for the verifier to learn $\pi(w) = p(w)$ and $\pi(w') = p(w')$ at two points $w$ and $w'$ of the verifier's choosing, given a witness state $\ket\psi$ sent by Merlin. 

We explain the verification procedure as we analyze the honest case (i.e., the case when there exists a $\pi$ such that the $\PCP$ verifier accepts with probability at least $c$). \ 
Let $a,b \in \{0,1\}^n$ be distinct points and let $c \in \F_q^n$ be independent of $a$ and $b$. \
Let $A$ be the unique affine plane that contains $a,b,$ and $c$. \ 
Let $C_{a,b,c}$ be the following function. \ 
For a vector $y \in \F_q^n$, if $y \in A$, then $C_{a,b,c}(y) = 0^n$. \
Otherwise, $a,b,c,$ and $y$ define a unique affine cube, which we denote by $B$. \
In this case, $C_{a,b,c}(y) = y' \in \F_q^n$ is a canonical representation of the point $y$, so that $(a,b,c,y)$ and $(a,b,c, y')$ define the same affine cube.

We describe the canonical representation in more detail. \
First, note that $y$ must be one of the $q^3 - q^2$ points in the cube $B$ that are not in the plane $A$ (otherwise $C_{a,b,c}(y) = 0^n$). \
There are many ways to pick a canonical representative $y'$. \ 
For example, of the $q^3 - q^2$ many points, one can have $C_{a,b,c}(y)$  output the point $y'$ with the fewest nonzero entries (and if there are ties, pick the one that comes first in lexicographic order). \
This ensures that all of the $q^3 - q^2$ points get mapped to the same canonical representative $y'$, which is crucial for our verification procedure. 
Recall that $q = O(n)$, so $C_{a,b,c}$ can be computed efficiently.

The honest $\PDQMA$ witness is 
\[ 
\frac{1}{\sqrt{q^n}} \sum_{z\in \mathbb{F}_q^n} \ket{z} \ket{p(z)}. 
\]
Given this witness, the $\mathsf{PDQP}$ verification procedure is as follows:

\begin{enumerate}
\item[(1)] Simulating the $\PCP$ verifier, choose two queries $w,w' \in \bin^n$. \ Pick a point $w'' \in \F_q^n$ uniformly at random.

\item[(2)] \label{step:ray} 
Map the witness to
\[
\frac{1}{\sqrt{q^n}}
 \sum_{z\in \mathbb{F}_q^n} \ket{z} \ket{p(z)} \ket{C_{w, w', w''}(z)}.
 \]

\item[(3)] \label{step:measure} Measure the $\ket{C_{w,w',w''}(z)}$ register in the usual collapsing way to obtain the outcome $y \in \F_q^n$. \ If the measurement outcome is $0^n$, reject. Let $B$ denote the affine cube containing $w, w'$, $w''$, and $y$, and let $A$ denote the affine plane containing $w$, $w'$, and $w''$. 

\item[(4)] Make $O(n^4)$ non-collapsing measurements of the $\ket{z}$ and $\ket{p(z)}$ registers. 

\item[(5)] 
If exactly the $q^3 - q^2$ points in $B \setminus A$ are obtained and the empirical distribution over these points is $O(1/n)$-close in total variation distance to the uniform distribution, continue. \ Otherwise, reject.

\item[(6)] If more than one $p(z)$ value was obtained for the same $z$, reject.

\item[(7)] Perform polynomial interpolation to obtain trivariate polynomial $u_B: \mathbb{F}_q^3 \rightarrow \mathbb{F}_q$ of degree at most $n$ that is consistent with $p$ on the measured points. \ If this interpolation fails, then reject.

\item[(8)] Calculate $\pi(w)=p(w)=u(0,0,0)$ and $\pi(w')=p(w')=u(0,1,0)$. \ Plug these responses into the $\mathsf{PCP}$ verifier, and accept if and only if it does.
\end{enumerate}

Let us analyze the verification procedure in more detail. \
Suppose, upon measuring the register $\ket{C_{w,w',w''}(z)}$ in Step 3, the verifier sees $y \in \F_q^n$. \
With probability $q^{2-n}$, we observe $y = 0^n$, because there are $q^2$ points on the affine plane $\{w + (w' - w)t_1 + (w'' - w)t_2\}_{t_1, t_2 \in \F_q}$. \ 
In this case, the verifier will reject (which is incorrect). \ 
Otherwise, with probability $1 - q^{2-n} = 1 - \exp(-\Omega(n))$, the verifier will see some canonical point $y \in \F_q^n$ so that the points $(w, w', w'', y)$ defines an affine cube $B$. 

The post-measurement state of the remaining two registers will then be in superposition over all points $z \in \F_q^n$ such that $z, w$, $w'$, and $w''$ define the same affine cube as $y$, $w$, $w'$, and $w''$. \
Recall that the affine cube $B$ is the set 
\[
B = 
\{ w + (y - w) t_1 + (w' - w) t_2 + (w'' - w) t_3 \}_{t_1,t_2,t_3 \in \F_q}, 
\]
and the affine plane $A \subseteq B$ containing $w, w',$ and $w''$ is the set 
\[
A = \{ w + (w' - w) t_1 + (w'' - w) t_2 \}_{t_1,t_2 \in \F_q}. 
\]
Observe that $z$ can be any of the $q^3 - q^2$ points in $B \setminus A$. \
In particular, 
\[
z = w + (y - w)t_1 + (w' - w) t_2 + (w'' - w) t_3
\]
for any $t_1 \in \F_q^*$ and $t_2,t_3 \in \F_q$. \
Therefore, our post-measurement state $\ket\phi$ can be expressed as
\[
\ket{\phi} = \frac{1}{\sqrt{q^3 - q^2}} \sum_{\substack{t_1 \in \F_q^*\\ t_2 \in \F_q\\t_3\in\F_q}} \ket{w + (y - w)t_1 + (w' - w) t_2 + (w'' - w) t_3} \ket{p(w + (y - w')t_1 + (w' - w) t_2 + (w'' - w) t_3)}.
\]
Define $u_B:\F_q^3 \to \F$ by $u(t_1, t_2,t_3) \coloneqq p(w + (y - w)t_1 + (w' - w) t_2 + (w'' - w)t_3)$, and notice that $u_B(0,0,0) = p(w)$ and $u_B(0,1,0) = p(w')$. \ Because $p$ is the multilinear extension of $\pi$, we also have that $p(w) = \pi(w)$ and $p(w') = \pi(w')$.

After Step 6, the verifier has collected $q^3 - q^2$ pairs $(z, p(z))_{z \in B \setminus A}$. \ 
Collecting these pairs is an instance of the coupon collector's problem, so $O(q^3 \log q) = O(n^3 \log n)$ many samples suffice to succeed with high probability. \ 
We take more samples, which will be relevant to the soundness of our protocol. \ 
With the $(z, p(z))$ pairs, the verifier runs polynomial interpolation to learn the polynomial $u_B$. \
We note that $q$ is chosen to be $\geq n+2$ to ensure that $q^3 - q^2$ pairs suffice for polynomial interpolation.  \
After learning $u_B$, the verifier has learned $\pi(w)$ and $\pi(w')$ as desired and will accept with the same probability as the $\PCP$ verifier. 

A crucial part of our verification procedure is that the verifier tests that $p$ is a low-degree polynomial. \ 
Because we picked a point $w'' \in \F_q^n$ uniformly at random, we ensure that the affine cube $B$ contains a random affine line, independent of the points $w$ and $w'$. \
Therefore the polynomial interpolation succeeds if and only if the truth table of $p$ matches a low-degree polynomial on a randomly chosen line $\ell$. \ 
Hence, by \cref{lemma:ldt}, $p$ must also be globally low degree. \ 
After taking into account the failures that can occur during the verification procedure, we conclude that the verifier accepts with probability at least $c - \exp(-\Omega(n))$.

We now analyze the soundness of our verification procedure. \ That is, suppose the $\PCP$ verifier will accept with probability at most $s$ for all possible proofs $\pi$. \ 
We will show that the $\PDQMA$ verifier accepts with probability at most $s$. \
The key insight for the soundness case is that, by deviating from the honest witness state above, Merlin only increases the probability that the polynomial interpolation will fail, causing Arthur to reject. \ 
In particular, the only way Merlin can cheat is to encode a truth table in $\ket\psi$ that is not degree $n$, but some function with much larger degree. \ 
However, the verifier will detect this with the lines-point low-degree test. 

Before going through the technical details, let us emphasize that Merlin does not know the points $w, w'$, and $w''$ that the verifier will select---these are chosen after Merlin commits to a witness state. \ 
Hence, Merlin must send a witness state $\ket\psi$ that passes all the checks in the verification procedure for all choices of $w, w',$ and $w''$ and no matter the random outcome $y$ the verifier observes in Step 3. 

Formally, Merlin can send an arbitrary state: 
\[
\ket{\psi} = \sum_{z\in\F_q^n, b \in \F_q} \alpha_{z,b} \ket{z}\ket{b}.
\]
The verifier maps the witness to 
\[
\sum_{z\in\F_q^n, b \in \F_q} \alpha_{z,b} \ket{z}\ket{b}\ket{C_{w, w',w''}(z)},
\]
and measures the last register. \ 
Suppose the measurement outcome is some $y \in \F_q^n$. \
If $y = 0^n$, the verifier rejects, but we will pessimistically assume this never happens. \
Suppose $y \neq 0^n$. \ 
As discussed previously, there are $q^3 - q^2$ points $z$ such that $z, w,w'$, and $w''$ define the same affine cube $B$ as $y, w,$ $w'$, and $w''$. \ 
These correspond to the points in $B \setminus A$ (recall that $A$ is the affine plane containing $w, w'$, and $w''$). \
Define $D \subseteq B \setminus A$ to be the points $z \in B \setminus A$ for which there exists at least one nonzero $\alpha_{z,b}$ for some $b \in \F_q$. \ 
The post-measurement state $\ket\phi$ is then
\[
\ket\phi = \sum_{z \in D, b \in \F_q} \widetilde{\alpha_{z,b}} \ket{z}\ket{b},
\]
where 
\[\widetilde{\alpha_{z,b}} = \frac{\abs{\alpha_{z,b}}}{\sqrt{\sum_{z \in D, b \in \F_q} \abs{\alpha_{z, b}}^2}}.\]
Note that we can assume without loss of generality that $\widetilde{\alpha_{z,b}} \in \mathbb{R}$ as the verifier only performs non-collapsing measurements on $\ket\phi$. \
In fact, in Step 4, the verifier's actions can be understood as drawing samples from a classical probability distribution where each $(z,b)$ pair has a probability of  $|\widetilde{\alpha_{z,b}}|^2$.

Recall the well-known fact that $\Theta(\frac{n + \log(1/\delta)}{\eps^2})$ samples are necessary and sufficient to learn a distribution to total variation distance at most $\eps$ with probability at least $1 -\delta$ (cf. \cite[Theorem 1]{Can20}). \
Hence, the $O(n^4)$ non-collapsing measurements in Step 5 suffice to learn the distribution over $(z,b)$ pairs in the support of $\ket\phi$ to total variation distance at most $O(1/n)$ with probability at least $1 - \exp(-\Omega(n))$. \
In particular, for Steps 4 through 6 to pass, $\ket\phi$ must be (approximately) uniformly supported on pairs $(z, b_z)$ for each $z \in B \setminus A$. \
(The verifier will immediately reject if any $z \in B\setminus A$ is paired with more than one $b \in \F_q$ value.) \
Because this must hold for any affine plane $A$ and affine cube $B$, we can deduce that Merlin is forced to send a state that is (approximately) uniform over pairs $(z, b_z)$ for every $z \in \F_q^n$.

Assuming these steps pass, Step 7 passes only if the observed pairs $(z, b_z)_{z \in B \setminus A}$ fit a degree-$n$ trivariate polynomial $u_B$. \
As discussed in the honest case, by selecting a $w'' \in \F_q^n$ uniformly at random, we guarantee that the affine cube $B$ contains a random line $\ell$, independent of the $w$ and $w'$ (this is precisely why we need a $3$-dimensional affine subspace). \ 
Therefore, the interpolation succeeding implies that the function encoded by Merlin matches a degree-$n$ polynomial on a randomly chosen line. \ 
By \cref{lemma:ldt}, we can conclude that Merlin indeed encoded a truth table for a degree-$n$ polynomial. \

If all of these steps pass, then the verifier can learn $\pi(w)$ and $\pi(w')$ as desired by evaluating $u_B(0,0,0)$ and $u_B(0,1,0)$. \
By plugging the values $\pi(w)$ and $\pi(w')$ into the $\PCP$ verifier, the $\PDQMA$ verifier accepts with probability at most $s$. \
\end{proof}

\section{\texorpdfstring{$\sf{QMA}$ and Hidden Variables}{QMA and Hidden Variables}}\label{sec:dqma}

We introduce and characterize the complexity class $\DQMA$, a variant of $\QMA$ where the verifier can perform $\DQP$ computations. \
Informally, $\DQP$ is like $\BQP$ with the ability to inspect the entire history of a hidden variable. \
For completeness, we begin this section by giving a formal definition of $\DQMA$. \
Then, as with $\PDQMA$, we prove that $\DQMA = \NEXP$. \ 
This result is more ``quantum'' than $\PDQMA = \NEXP$ (\cref{thm:pdqma=nexp}) because the $\DQP$ verifier will use quantum circuits (as opposed to merely computational basis measurements). \

\subsection{\texorpdfstring{The Complexity Class $\DQMA$}{The Complexity Class DQMA}}
We give a formal definition of $\DQMA$. \
To do so, we must recall a few definitions related to the class $\DQP$ (see \cite{aar:qchv} for more detail about this class).

We begin by defining hidden-variable theories. \
To aid intuition, one can think of hidden-variable theories like standard quantum mechanics where there is a state vector that evolves unitarily. \ 
However, there is also a deeper ``hidden variable'' in some definite state (i.e., not in superposition) that evolves stochastically in a manner determined by the state vector and the state vector's unitary evolution. \ 
In the series of definitions that follow, we are building up to defining a model of computation where one evolves a state vector unitarily, and then (at the end of the computation) can inspect which states the hidden variable was in at each step of the computation. 

\begin{definition}[Hidden-variable theory]
    A hidden-variable theory is a family of functions $\{S_d\}_{d \in \N}$, where each $S_d$ maps a $d$-dimensional mixed state $\rho$ and a $d \times d$ unitary matrix $U$ onto a singly stochastic matrix $S_d(\rho, U)$.
\end{definition}

That is, we take a hidden-variable theory to be a function that maps the unitary evolution of a state to a stochastic matrix that evolves one probability distribution to another. \ 
Conditioned on a hidden variable being in a state $\ket{j}$, $(S)_{ij}$ is the probability that a hidden variable transitions to the $\ket{j}$. 

Aaronson \cite{aar:qchv} defined a number of axioms a hidden-variable theory could satisfy. \ 
We require three of these axioms to define $\DQP$: the marginalization axiom, the indifference axiom, and the robustness axiom. \ 
In the following definitions, $S$ denotes the hidden-variable theory, $\rho$ a $d$-dimensional quantum state, and $U$ a $d \times d$ unitary matrix. \
Each axiom must hold for all $d \in \N$.

\begin{definition}[Marginalization axiom]\label{def:marg-axiom}
The marginalization axiom says that for all $j \in \{1, \dots, d\}$, 
\[
\sum_i (S)_{ij} (\rho)_{ii} = (U \rho U^\dagger)_{jj}.
\]
\end{definition}

In words, the marginalization axiom says that the hidden-variable theory should make predictions that are consistent with quantum mechanics. 

\begin{definition}[Indifference axiom]\label{def:indiff-axiom}
For a matrix $M \in \mathbb{C}^{d \times d}$, let a block be a subset $B \subseteq \{1, \ldots, d\}$ such that $(M)_{ij} = 0$ for all $(i,j)$ such that $i \in B$, $j \notin B$ and $i \notin B$, $j \in B$. \  
The indifference axiom says that, the stochastic matrix $S(\rho, U)$ must have the same blocks as $U$.
\end{definition}

Physically, the indifference axiom is saying the following. \ Given any quantum state $\rho$ in a tensor product space $\calH_A \otimes \calH_B$ and any unitary $U$ acting nontrivially only on $\calH_A$, the stochastic matrix $S(\rho, U)$ acts nontrivially only on $\calH_A$ as well.

Finally, we state the robustness axiom, for which we need the following notation. \ Let $P(\rho, U)$ be the matrix of joint probabilities whose $(i,j)$ entry is $(P)_{ij} \coloneqq (S)_{ij}(\rho)_{ii}$. 

\begin{definition}[Robustness axiom]\label{def:robust-axiom}
    Let $\widetilde{\rho}$ and $\widetilde{U}$ be perturbations of $\rho$ and $U$, respectively, and, for a matrix $M$, let $\norm{M}_{\infty} \coloneqq \max_{i,j} \abs{(M)_{ij}}$. \
    The robustness axiom says that, for all polynomials $p$, there should exist a polynomial $q$ such that, 
    \[
       \norm{P(\widetilde{\rho}, \widetilde{U} ) - P\left(\rho, U \right)}_{\infty} \leq \frac{1}{p(d)},
    \]
    whenever $\norm{\widetilde{\rho} - \rho}_{\infty} \leq 1 / q(d)$ and $\norm{\widetilde{U} - U}\leq 1/q(d)$. 
\end{definition}

The robustness axiom is necessary to prove that the class $\DQP$ is not sensitive to the choice of gate set defining the class. 

Next, we define the history of a hidden variable. 

\begin{definition}[Hidden variable history]\label{def:hidden-variable-history}
Let $\ket{\psi_{\rm{init}}}$ be an $n$-qubit quantum state, and let $U \coloneqq U_T\cdots U_1$ be an $n$-qubit, depth-$T$ quantum circuit, where $U_1, \ldots, U_T$ denote each layer of the quantum circuit $U$. \ 
The history of a hidden variable is a sequence $H = (v_0, \ldots, v_T)$ of basis states, where $v_t$ is the state of the hidden variable immediately after the layer $U_t$ of the circuit is applied. \ 
Given a hidden-variable theory $\calT \coloneqq \{S_d\}_{d \in \N}$, we obtain a probability distribution over hidden variable histories $\Omega\left(\calT, U, \ket{\psi_{\rm{init}}}\right)$ via the stochastic matrices 
\[
S\left(\ket{\psi_{\rm{init}}}, U_1 \right),\,S\left(U_1\ket{\psi_{\rm{init}}}, U_2 \right), \ldots, S\left(U_{T-1} \cdots U_1 \ket{\psi_{\rm{init}}}, U_T \right).
\]
\end{definition}

We can now define the complexity class $\DQMA$. 

\begin{definition}
$\mathsf{DQMA}(c,s)$ is the class of languages $L\subseteq \bin^*$ for which there exists a $\mathsf{DQP}$ verifier $\verifier$ (to be defined shortly) such that, for all $x\in\bin^*$:
\begin{itemize}
\item \textbf{Completeness:} If $x\in L$ then there exists a witness state $\ket{\phi}$, on $\operatorname{poly}(n)$ qubits, such that $V(x,\ket{\phi})$ accepts with probability at least $c$.
\item \textbf{Soundness:} If $x\not\in L$ then $V(x,\ket{\phi})$ accepts with probability at most $s$ for all witness states $\ket{\phi}$.
\end{itemize}
A $\mathsf{DQP}$ verifier is defined as follows. \
Let $\calT$ be a hidden-variable theory satisfying the marginalization, indifference, and robustness axioms (\cref{def:marg-axiom,def:indiff-axiom,def:robust-axiom}), and let $C_x$ (depending on the input $x$) be a $\mathsf{P}$-uniform quantum circuit comprised of gates from any finite gate set that is universal for $\BQP$. \ 
A $\DQP$ verifier is a deterministic classical Turing machine that is allowed to draw one sample from the distribution $\Omega\left(\calT, C_x, \ket{\phi}\ket{0^{p(n)}} \right)$ (\cref{def:hidden-variable-history}), where $\ket{\phi}$ is the witness and $p$ is a polynomial. 
\end{definition}

\subsection{\texorpdfstring{$\DQMA = \NEXP$ and $\mathsf{DQP/qpoly} = \ALL$}{DQMA = NEXP and DQP/qpoly = ALL}}
We conclude this section by proving $\DQMA = \NEXP$ and $\mathsf{DQP/qpoly} = \ALL$. \
Recall that in the proof of $\PDQMA = \NEXP$ (\cref{thm:pdqma=nexp}), the verifier uses non-collapsing measurements to sample $O(n \log n)$ values of a multivariate polynomial along an affine line of one's choice and then does interpolation. \
To prove $\DQMA = \NEXP$, we use the same verification procedure, except we use the history of a hidden variable to collect the samples in place of non-collapsing measurements. \ 
Hence, to prove $\DQMA = \NEXP$ it suffices to explain how we collect these samples with the history of a hidden variable. \
We achieve this by generalizing the ``juggle subroutine" due to Aaronson \cite[Section VII]{aar:qchv}. 

\begin{lemma}[{Juggle subroutine \cite[Section VII]{aar:qchv}}]\label{lemma:original-juggle}
    Suppose we have an $\ell$-qubit state 
    \[\frac{\ket{a} \pm \ket{b}}{\sqrt{2}},\] where $\ket{a}$ and $\ket{b}$ are unknown basis states. \
    Given a single copy of the state, the juggle subroutine (a $\DQP$ algorithm) can efficiently learn both $a$ and $b$ with success probability at least $1 - e^{-\ell}$. 
\end{lemma}

We generalize this algorithm to work on states that are an equal superposition over polynomially many strings by reducing to the case of an equal superposition over two strings. \
Before explaining the generalization, we first must define (and give a simple fact about) pairwise independent families of hash functions, which are used in our reduction.

\begin{definition}[Pairwise independent family of hash functions]
A family of hash functions $\calH = \{h : \{0,1\}^\ell \to R\}$ is called pairwise independent if $\forall$ $x\neq y \in \{0,1\}^\ell$ and $\forall a_1,a_2 \in R$, we have 
\[
\Pr[h(x) = a_1 \wedge h(y) = a_2] = \frac{1}{\abs{R}^2}.
\]
\end{definition}

\begin{fact}\label{fact:one-collision}
    Let $S \subseteq \{0,1\}^\ell$ be a subset, 
    and let $\calH \coloneqq \{h : \{0,1\}^\ell \to R\}$ be a family of pairwise independent hash functions such that $2 \abs{S}-4 \leq \abs{R} \leq  3 \abs{S} - 3$. \ 
    Then, for any fixed $x \in S$, the probability that $x$ collides with exactly one other $y \in S$ is at least $\frac{1}{6}$.
\end{fact}
\begin{proof}
   Let $h \in \calH$ be chosen uniformly at random, and let $x \in S$ be some fixed element. \ 
    The probability that exactly one other element collides with $x$ is 
    \[
     \frac{\abs{S}-1}{\abs{R}} \cdot \left( 1 - \frac{1}{\abs{R}} \right)^{\abs{S} -2}.
    \]
    We lower bound this quantity.
    \begin{align*}
    \frac{\abs{S}-1}{\abs{R}} \cdot \left( 1 - \frac{1}{\abs{R}} \right)^{\abs{S} -2}
    &\geq \frac{\abs{S}-1}{\abs{R}} \cdot \left( 1 - \frac{\abs{S} -2}{\abs{R}} \right) \\
    &\geq \frac{1}{3} \cdot \left( 1 - \frac{\abs{S}-2}{\abs{R}} \right) \\
    &\geq \frac{1}{6}. 
    \end{align*}
    The first inequality follows from Bernoulli's inequality. \ The second and third inequalities use the fact that $2\abs{S} -4 \leq \abs{R} \leq 3\abs{S} -3$.
\end{proof}

We now give the generalized juggle subroutine.

\begin{lemma}[{Generalized juggle subroutine}]\label{lemma:generalized-juggle}
    Suppose we have an $\ell$-qubit state 
    \[
   \frac{1}{\sqrt{\abs{S}}} \sum_{x\in S} \ket{x},
    \] 
    where $S \subseteq \{0,1\}^\ell$ is an unknown subset of $\abs{S} \leq \poly(\ell)$ basis states. \ 
    Given a single copy of the state, the generalized juggle subroutine (a $\DQP$ algorithm) can efficiently learn $S$ with success probability at least $1 - e^{-\ell}$.
\end{lemma}
\begin{proof}
We begin with the state of the form
\[
\ket{\psi} = \frac{1}{\sqrt{\abs{S}}} \sum_{x\in S} \ket{x}.
\]
We perform a procedure that involves applying transformations to $\ket\psi$ and then inverting them to get back to $\ket\psi$, in an attempt to ``dislodge'' the hidden variable from whichever basis state it's currently sitting in, and get it into a different, uniformly random one. \  
Each time we do this, we have a $1/\poly(\ell)$ probability of success. \ 
Importantly, since there's no penalty for failure, we can repeat this procedure $\poly(\ell)$ times, and then with overwhelming probability, the hidden variable will have visited every basis state $\ket{x}$ for $x \in S$. \
Therefore, one will learn $S$ upon observing the hidden-variable history. 

The procedure works as follows. \
First, choose a random hash function $h$ (with a range satisfying the conditions in \cref{fact:one-collision}) from some pairwise-independent family, and then map each $\ket{x}$ to $\ket{x}\ket{h(x)}$. \
Let $y$ be the current state of the hidden variable. \ 
By \cref{fact:one-collision}, with probability at least $\frac{1}{6}$, there's \emph{exactly} one other basis state $z\neq y$ such that $h(z) = h(y)$, and this $z$ is uniformly random. \  
If that happens, then because of the indifference axiom (which says that, if we don't touch the $h$-register, then the hidden variable will never move between $h$-values), we've reduced to the problem handled by the original juggle subroutine. \  
In particular, we can now run the original juggle subroutine on the first register (where the hidden variable is). \
\cref{lemma:original-juggle} tells us that, with probability at least $1 - e^{-\ell}$, the hidden variable moves from $y$ to $z$. \ 
Finally, we uncompute $h$, leaving us with our original state $\ket\psi$. \ 
Overall, the inner loop moves the hidden variable from $y$ to a uniformly random $z$ with probability at least $\frac{1}{6} - \frac{e^{-\ell}}{6} \geq 1/\poly(\ell)$.\footnote{We note that there is some chance that the hidden variable stays put or moves to somewhere other than $z$, but that's OK too, since our ultimate goal is for the hidden variable to visit every possible state in $S$. In any case, we keep repeating.} \
Since there is no penalty for failure, we can repeat this procedure $2 \ell^2$ times to ensure that with probability at least $1 - e^{-\ell}$, the hidden variable was successfully moved to a uniformly random basis state. \
Finally, since we visit a uniformly random state with high probability, we can visit every state with high probability by repeating this entire procedure a polynomial number of times.  
\end{proof}

We are now ready to prove the main theorem of this section. Namely, that giving Arthur access to hidden-variable histories blows up the power of $\QMA$ to $\NEXP$.

\begin{theorem}\label{thm:dqma=nexp}
    $\DQMA = \NEXP.$
\end{theorem}
\begin{proof}
It is clear that $\DQMA \subseteq \NEXP$. \ 
By \cref{mipnexp}, we can complete the proof by showing $\MIP \subseteq \DQMA$.

The verification procedure and the honest witness sent by Merlin are the same as in the proof of \cref{thm:pdqma=nexp}. \
Suppose we have a yes-instance, and Merlin sends the honest witness with the form:
\[ \frac{1}{q^n} \sum_{z\in \mathbb{F}_q^n} \ket{z} \ket{p(z)}.\]
We must explain how to use the history of a hidden variable in lieu of non-collapsing measurements. \
After the verifier makes \emph{collapsing} measurements, they must learn the support of the post-measurement state. \
To do this, the verifier runs the generalized juggle subroutine (\cref{lemma:generalized-juggle}). \ 
After that, the verifier can do polynomial interpolation with just classical computation (or reject if there is insufficient data to do the interpolation).

The only difference in the verification procedure is that the data for polynomial interpolation is collected via the generalized juggle subroutine instead of non-collapsing measurements. \
Therefore, he completeness and soundness of this procedure follow in the same way as for \cref{thm:pdqma=nexp}. \
In particular, if Merlin deviates from the honest witness state, then he can only hurt his success probability by causing the polynomial interpolation step to fail. \
\end{proof}

Finally, we remark that our generalized juggle subroutine can be used to prove $\mathsf{DQP/qpoly} = \mathsf{ALL}$, complementing Aaronson's result that $\mathsf{PDQP/qpoly} = \mathsf{ALL}$ \cite{aar:pdqpqpoly} and Raz's result that $\mathsf{QIP(2)/qpoly} = \mathsf{ALL}$ \cite{raz:all}. \
Similar to \cref{thm:dqma=nexp}, this result is more ``quantum'' than Aaronson's or Raz's, because the generalized juggle subroutine requires quantum computation. 

\begin{theorem}\label{thm:dqp/qpoly=all}
    $\mathsf{DQP/qpoly} = \mathsf{ALL}.$
\end{theorem}
\begin{proof}
    The advice state and verification procedure are the same as in \cite[Theorem 2]{aar:pdqpqpoly}, except we replace the non-collapsing measurements with the generalized juggle subroutine in the same manner described in the proof of \cref{thm:dqma=nexp}.
\end{proof}
    
\section{Open Problems}\label{sec:open-problems}

Is $\PDQP \subseteq \DQP$? Aaronson et al. \cite{abfl} give intuition for why this containment ought to be true, but it remains an open problem. \
Proving $\PDQP \subseteq \DQP$ (combined with \cref{thm:pdqma=nexp}) immediately implies $\DQMA = \NEXP$, simplifying our proof in \cref{thm:dqma=nexp}. \ 
We also remark that improving the upper bound $\DQP \subseteq \EXP$ remains an interesting open problem. 

We now know that modifying $\QMA$ by giving the verifier access to non-collapsing measurements, hidden-variable histories, or non-negative witnesses will cause $\QMA$ to ``explode'' in power to $\NEXP$. \
Is it possible to replace the verifier with some variant that does \emph{not} lead to $\NEXP$? \
Having access to such variants may find applications in proving better bounds on $\mathsf{QMA(2)}$.

\section{Acknowledgments}
SCM thanks David Dechant for useful discussions. SG, VI, and RR thank Joshua Cook, Siddhartha Jain, and Kunal Marwaha for helpful conversations.

SA is supported by a Vannevar Bush Fellowship by the US Department of Defense, the Berkeley NSF-QLCI CIQC Center, a Simons Investigator Award, and the Simons \textquotedblleft It from Qubit\textquotedblright\ collaboration. SG and RR are supported via SA from the same funding sources. 
VI is supported by an NSF Graduate Research Fellowship.
SCM is supported by the European Union's Horizon 2020 program (NEQASQC) and a gift from Google Quantum AI.

\bibliographystyle{alphaurl}
\bibliography{thesis}

\newpage

\end{document}